\documentclass[12pt]{article}
\usepackage{graphicx}
\usepackage{amsmath,amsthm,amssymb,enumerate}%, esint}
\usepackage{euscript,mathrsfs}
\usepackage{nicefrac}
\usepackage{bm}
\usepackage{cite}
\usepackage{morefloats}
\usepackage[citecolor=blue,colorlinks=true]{hyperref}
\usepackage[left=2cm,right=2cm,top=3.5cm,bottom=3.5cm]{geometry}
\usepackage{color}

\usepackage{soul}

\catcode`\@=11 \@addtoreset{equation}{section}

\catcode`\@=12

\allowdisplaybreaks

\newtheorem{Theorem}{Theorem}[section]
\newtheorem{Proposition}[Theorem]{Proposition}
\newtheorem{Lemma}[Theorem]{Lemma}
\newtheorem{Corollary}[Theorem]{Corollary}

\theoremstyle{definition}
\newtheorem{Definition}[Theorem]{Definition}

\newtheorem{Remark}[Theorem]{Remark}

\newcommand{\bfom}{\mbox{\boldmath $\varpi$}}
\newcommand{\bTheorem}[1]{
\begin{Theorem} \label{T#1} }
\newcommand{\eT}{\end{Theorem}}

\newcommand{\bProposition}[1]{
\begin{Proposition} \label{P#1}}
\newcommand{\eP}{\end{Proposition}}

\newcommand{\bLemma}[1]{
\begin{Lemma} \label{L#1} }
\newcommand{\eL}{\end{Lemma}}

\newcommand{\bCorollary}[1]{
\begin{Corollary} \label{C#1} }
\newcommand{\eC}{\end{Corollary}}

\newcommand{\bRemark}[1]{
\begin{Remark} \label{R#1} }
\newcommand{\eR}{\end{Remark}}

\newcommand{\bDefinition}[1]{
\begin{Definition} \label{D#1} }
\newcommand{\eD}{\end{Definition}}

\newcommand{\Del}{\Delta_x}

\DeclareMathOperator{\supp}{supp}

\newcommand{\bfphi}{\boldsymbol{\varphi}}

\newcommand{\bfeta}{\boldsymbol{\eta}}

\newcommand{\bFormula}[1]{
\begin{equation} \label{#1}}
\newcommand{\eF}{\end{equation}}

\newcommand{\Ov}[1]{\overline{#1}}

\newcommand{\vr}{\varrho}

\newcommand{\vu}{\bm{u}}
\newcommand{\vm}{\bm{m}}

\newcommand{\vn}{\bm{n}}

\newcommand{\vF}{\bm{F}}
\newcommand{\pat}{\partial_t}
\newcommand{\vV}{\vu^S}

\newcommand{\vomega}{\boldsymbol{\omega}}

\newcommand{\veta}{\boldsymbol{\eta}}

\newcommand{\vc}[1]{{\bm #1}}

\newcommand{\Div}{{\rm div}_x}
\newcommand{\Grad}{\nabla_x}

\newcommand{\dx}{\,{\rm d} {x}}

\newcommand{\dt}{\,{\rm d} t }

\newcommand{\dxdt}{\dx \ \dt}
\newcommand{\intO}[1]{\int_{\Omega} #1 \ \dx}

\renewcommand{\vm}{\bm{m}}
\renewcommand{\vu}{\bm{u}}

\def\softd{{\leavevmode\setbox1=\hbox{d}%
          \hbox to 1.05\wd1{d\kern-0.4ex{\char039}\hss}}}
\definecolor{Cgrey}{rgb}{0.85,0.85,0.85}
\definecolor{Cblue}{rgb}{0.50,0.85,0.85}
\definecolor{Cred}{rgb}{1,0,0}
\definecolor{fancy}{rgb}{0.10,0.85,0.10}

\newcommand\Cbox[2]{%
    \newbox\contentbox%
    \newbox\bkgdbox%
    \setbox\contentbox\hbox to \hsize{%
        \vtop{
            \kern\columnsep
            \hbox to \hsize{%
                \kern\columnsep%
                \advance\hsize by -2\columnsep%
                \setlength{\textwidth}{\hsize}%
                \vbox{
                    \parskip=\baselineskip
                    \parindent=0bp
                    #2
                }%
                \kern\columnsep%
            }%
            \kern\columnsep%
        }%
    }%
    \setbox\bkgdbox\vbox{
        \color{#1}
        \hrule width  \wd\contentbox %
               height \ht\contentbox %
               depth  \dp\contentbox
        \color{black}
    }%
    \wd\bkgdbox=0bp%
    \vbox{\hbox to \hsize{\box\bkgdbox\box\contentbox}}%
    \vskip\baselineskip%
}

\date{}

\begin{document}

%%%%%%%%%%%%%%%%%%%%%%%%%%%%%%%%

\title{On the motion of rigid bodies in a perfect fluid}

\author{Eduard Feireisl\thanks{The research of E.F. and V.M. leading to these results has received funding from the
Czech Sciences Foundation (GA\v CR), Grant Agreement
18--05974S. The Institute of Mathematics of the Academy of Sciences of
the Czech Republic is supported by RVO:67985840.} $^{\spadesuit}$ $^\clubsuit$
\and V\' aclav M\' acha $^{\spadesuit}$
}

\date{\today}

\maketitle

\bigskip

\centerline{$^{\spadesuit}$ Institute of Mathematics of the Academy of Sciences of the Czech Republic}
\centerline{\v Zitn\' a 25, CZ-115 67 Praha 1, Czech Republic}
\centerline{feireisl@math.cas.cz}

\bigskip

\centerline{$^\clubsuit$ Institute of Mathematics, TU Berlin}
\centerline{Strasse des 17. Juni, Berlin, Germany}

\begin{abstract}

We consider the problem of motion of several rigid bodies immersed in a perfect compressible fluid. Using the method of convex 
integration we establish the existence of infinitely many weak solutions with {\it a priori} prescribed motion of rigid bodies.
In particular, the dynamics is completely \emph{time--reversible} at the motion of rigid bodies  although the solutions 
comply with the standard entropy admissibility criterion.

\end{abstract}


%\tableofcontents

\section{Introduction}
\label{i}

The motion of one or more rigid objects immersed in a fluid is an example of \emph{fluid--structure interaction problem} in continuum 
mechanics. We focus on the case of perfect (inviscid) fluid contained in a bounded cavity $\Omega \subset R^d$, $d=2,3$. As is well--known, neglecting completely the effect of viscosity leads to unphysical conclusions among which the best known is 
the celebrated D'Alembert paradox: Both drag and lift vanish in a potential inviscid incompressible fluid flow. We consider a more realistic situation of a perfect \emph{compressible} fluid and show that the initial--value problem 
for the associated Euler system
is essentially ill--posed in the class of weak solutions. As weak solutions are indispensable in gas dynamics, where shock waves develop in finite time, the result suggests the model based on the Euler system calls for a thorough revision. 

\subsection{Rigid body motion}

The position of rigid bodies at a time $t \in [0,T]$ is represented by compact sets $B_i(t) \subset \Omega \subset R^d$, $i=1,\dots,N$.
The mass distribution in the bodies is determined by the density $\vr^S = \vr^S(t,x)$, 
\begin{equation} \label{i1}
\vr^S(t,x) = \left\{ \begin{array}{l} \vr^S_i (t,x) \geq 0, \ \mbox{if}\ x \in B_i(t), \ \vr^S_i \not\equiv 0 
\ \mbox{in}\ B_i(t)\\ \\ 
0 \ \mbox{otherwise.} 
\end{array} \right. 
\end{equation}  
We denote 
\begin{equation} \label{i2}
\begin{split}
m_i &\equiv \int_{B_i(t)} \vr^S \ \dx \ \mbox{-- the total mass}, \\  
x_{B_i} (t) &= \frac{1}{m_i} \int_{B_i(t)} \vr^S(t,x) x \ \dx \ \mbox{-- the barycenter of the body}\ B_i
\ \mbox{at a time}\ t \in [0,T].
\end{split}
\end{equation}

The motion of the rigid bodies is described through a velocity field 
$\vu^S = \vu^S(t,x):[0,T] \times \Ov{\Omega} \to R^d$, 
\begin{equation}\label{i3}
\Div \vV = 0,\ \vu^S \cdot \vc{n}|_{\partial \Omega} = 0,\  
\vu^S(t,x) = \veta_i(t) - \vomega_i(t)\times (x - x_{B_i}(t)) \ \mbox{if}\ x \in B_i(t).
\end{equation}
The field $\vu^S$ generates a flow map
$\vc{X}$, 
\[
\frac{{\rm d}}{{\rm d}t} \vc{X}(t, \vc{X}_0) = \vu^S(t, \vc{X}), \ \vc{X}(0, \vc{X}_0) = \vc{X}_0 \in \Ov{\Omega}.
\]
In accordance with \eqref{i3}, 
\begin{equation} \label{i4}
\vc{X}(t, \cdot): B_i \to B_i(t) \ \mbox{is an isometry for any}\ t \geq 0, \ i = 1,\dots, N.
\end{equation}
In addition, we set 
\[
\vr^S(t,\vc{X}(t,x)) = \vr^S(0,x) \ \mbox{for any}\ x \in \Ov{\Omega}, 
\]
meaning $\vr^S$ satisfies the equation of continuity 
\begin{equation} \label{i5}
\partial_t \vr^S + \Div (\vr^S \vu^S ) = 0 
\end{equation}
in the sense of distributions.

The time evolution of the velocity field is governed by Newton's second law. If $d=3$, introducing the inertial tensor 
\[
\mathbb{J}^i \cdot \vc{a} \cdot \vc{b} \equiv 
\int_{B^i(t)} \vr^S(t,x) \Big[ \vc{a} \times \left( x - x_{B_i(t)} \right) \Big] \cdot \Big[ \vc{b} \times \left( x - x_{B_i(t)} \right) \Big] \ \dx,
\]
we can write the momentum equation in the form (see Galdi \cite{galdi} or Houot, San Martin, and Tucsnak \cite{HoMaTu})
\begin{equation} \label{i6}
m_i \frac{{\rm d}}{{\rm d}t} \veta_i(t) =  - \int_{\partial B_i(t)} \mathbb{T} \cdot \vc{n} \ {\rm d}S_x + 
\int_{B_i(t)} \vr^S \vc{g} \ \dx,
\end{equation}
\begin{multline} \label{i7}
\mathbb{J}^i \cdot \frac{{\rm d}}{{\rm d}t} \vomega_i(t) = [ \mathbb{J}^i \cdot \vomega_i(t)] \times \vomega_i(t)\\ 
-\int_{\partial B_i(t)} \Big[ x - x_{B_i(t)} \Big] \times \Big[ \mathbb{T} \cdot \vc{n} \Big] \ {\rm d}S_x +
\int_{B_i(t)} \vr^S(t,x) \Big[ x - x_{B_i(t)} \Big] \times \vc{g} \ \dx.
\end{multline}
Here $\mathbb{T}$ is the total Cauchy stress acting on the body and $\vc{g}$ is a given external body force. In the real world applications, $\vc{g}$ is the gravitational force. If the bodies are immersed in a perfect fluid, the tensor 
$\mathbb{T}$ reduces to 
\begin{equation} \label{i8} 
\mathbb{T} = - p_F \mathbb{I},
\end{equation}
where $p_F$ is the fluid pressure.

In the case $d = 2$ the proper equations can be deduced by using appropriate projection. The momentum equation \eqref{i6} remains the same, however, the inertia tensors $\mathbb J^i$ have only one component and this can be computed as
\begin{equation}
\mathbb J^i = \int_{B^i(t)} \vr^S(t,x) |x-x_{B_i(t)}|^2 \ \dx.
\end{equation}
Then, rotation is represented just by a scalar quantity $\omega$ and \eqref{i7} is replaced by
\begin{equation}
\mathbb J^i \frac{\rm d}{{\rm d}t} \omega_i(t) = -\int_{\partial B_i(t)} (x-x_{B_i(t)})\cdot (\mathbb T \vc{n})^\bot\ {\rm d}S_x + \int_{B_i(t)} \vr^S(t,x) (x-x_{B_i(t)}) \cdot \vc{g}^\bot \ \dx.
\end{equation}
See also Ortega, Rosier, and Takahashi \cite{OrRoTa}.

\subsection{Fluid motion}

Neglecting the thermal effect, we suppose that the time evolution of the fluid density $\vr^F = \vr^F(t,x)$ 
and the fluid velocity $\vu^F = \vu^F(t,x)$ is governed by the barotropic Euler system
\begin{equation} \label{i9}
\partial_t \vr^F + \Div (\vr^F \vu^F ) = 0,
\end{equation} 
\begin{equation} \label{i10}
\partial_t (\vr^F \vu^F) + \Div (\vr^F \vu^F \otimes \vu^F) + \Grad p (\vr^F) = \vr^F \vc{g}.
\end{equation}

The equations are satisfied in the fluid domain 
\[
Q_F = \left\{ t \in (0,T),\ x \in \Omega_F (t)\ \Big| \ \Omega_F(t) = \Omega \setminus \cup_{i = 1}^N B_i(t) \right\}.
\]
Finally, we impose the impermeability boundary conditions 
\begin{equation} \label{i11}
\vu^F \cdot \vc{n} = \vu^S \cdot \vc{n} |_{\partial \Omega_F},
\end{equation}
where $\vu^S$ is the velocity field governing the motion of the bodies, cf. \eqref{i3}.

\subsection{Compatibility}

The motion of the rigid bodies is driven by the surrounding perfect fluid, if 
the Cauchy stress $\mathbb{T}$ satisfies 
\begin{equation} \label{i12}
\mathbb{T} = - p_F \mathbb{I},\ \mbox{where}\  p_F = p(\vr^F) \ \mbox{on} \ \partial B_i, \ i = 1,\dots, N.
\end{equation}
For relation \eqref{i10} to make sense, a notion of boundary trace of the fluid density $\vr^F$ is necessary. As the latter may not 
be continuous, we consider a larger class 
\[
\vr^F(t, \cdot) \in BV (\Omega_F(t)) \ \mbox{for}\ t \geq 0
\]
where at least one--sided traces are available as soon as $\partial B_i$ is Lipschitz, cf. 
Evans and Gariepy \cite{EVGA}.

\bigskip

As is well known, smooth solutions of the Euler system develop singularities in a finite time for a fairly general class of initial data.
Our goal is therefore to consider the fluid--structure interaction problem in the framework of weak solutions. 
We show that the corresponding initial--value problem admits infinitely many solutions
for any {\it a priori} given admissible motion of the rigid objects. By admissible motion we mean the motion of the rigid objects 
governed by the velocity field satisfying \eqref{i6}, \eqref{i7} with
\[
\mathbb{T} = - p_F = - p(\vr^F) \ \mbox{for a certain}\ \vr^F \in C^1(\Ov{Q}_F),\ \int_{\Omega_F(t)} \vr^F (t,x) \ \dx 
= m^F = {\rm const}\ \mbox{for any}\ t \in [0,T].
\] 
Although, apparently, not any given motion is admissible - the classical example is a ball $B$ with homogeneous density 
distribution for which the barycenter coincides with the geometric center and there is no way how to control the rotational 
velocity $\vomega$ - our result implies certain reversibility of the rigid body evolution. In addition, there are always infinitely 
many solutions (density and velocity od the fluid) giving rise to the same body motion. The
result is proved by the abstract machinery 
of convex integration developed in \cite{Fei2016}. The crucial observation 
is that the ``incompressible'' convex integration technique developed in earlier work of De Lellis and Sz\' ekelyhi 
\cite{DelSze3} can be adapted to the compressible Euler system with an {\it a priori} given density, cf. \cite{Fei2016}.

The paper is organized as follows. In Section \ref{M}, we introduce the concept of weak solution to the fluid--structure interaction problem and state our main result. In Section \ref{R}, we reformulate the problem to fit the abstract framework of convex integration.
In Section \ref{C}, we apply the method of convex integration within the framework developed in \cite{Fei2016} to show the existence of 
infinitely many solutions for given rigid body motion. The paper is concluded by a discussion about physically relevant solutions 
in Section \ref{P} and concrete examples of admissible rigid bodies motion in Section \ref{pres.con}.

\section{Weak formulation, main results}

\label{M}

We start by introducing the concept of weak solution to the fluid--structure interaction problem \eqref{i1}--\eqref{i12}. For the sake of simplicity, we omit the effect of volume forces setting $\vc{g} \equiv 0$. 

\begin{Definition}[{\bf Weak solution}] \label{MD1}

The velocity quantity $[\vr, \vu, \{ B_i \}_{i=1}^N]$ is a \emph{weak solution} of the problem \eqref{i1}--\eqref{i12} 
with the initial data $[\vr_0, \vm_0, \{ B_{i,0} \}_{i = 1}^N ]$
if the following holds:

\begin{itemize}

\item {\bf Integrability.}

\[
\vr \in L^\infty((0,T) \times \Omega),\ \vu \in L^\infty((0,T) \times \Omega; R^d),
\]
\[
0 < \underline{\vr} \leq \vr(t,x) \leq \Ov{\vr}\ \mbox{for a.a}\ (t,x) \times \Omega.
\]

\item {\bf Mass conservation.}
The integral identity 
\begin{equation} \label{M1}
\int_0^T \intO{ \Big[ \vr \partial_t \varphi + \vr \vu \cdot \Grad \varphi \Big] } \dt = - \intO{ \vr_0 \varphi (0, \cdot) }
\end{equation}
holds for any $\varphi \in C^1_{\rm loc}([0,T) \times \Ov{\Omega})$.

\item {\bf Compatibility.}
There is a velocity field $\vu^S \in W^{1,\infty}([0,T] \times \Omega ; R^d)$, 
\[
\Div \vV = 0,\ \vu^S \cdot \vc{n}|_{\partial \Omega} = 0,\  
\vu^S(t,x) = \vu(t,x) = \veta_i(t) - \vomega_i(t)\times (x - x_{B_i}(t)) \ \mbox{if}\ x \in B_i(t),
\]
and a density $\vr^F \in L^\infty(Q_F)$ such that
\[
\vr = \vr^F \ \mbox{in}\ Q_F,\ 
\vr^F(t, \cdot) \in BV (\Omega_F(t)) \ \mbox{for a.a.}\ t \in (0,T).
\]

\item
{\bf Momentum balance.}
For
\[
\vr(t,x) = \left\{ \begin{array}{l} \vr^F(t,x) \ \mbox{if}\ t \in [0,T], \ x \in (\Omega \setminus \cup_{i = 1}^N B_i(t)),\\ \\ 
\vr^S(t,x) \ \mbox{if}\ t \in [0,T], \ x \in \cup_{i = 1}^N B_i(t), 
\end{array}
\right.
\]  
\[
\vu(t,x) = \left\{ \begin{array}{l} \vu^F(t,x) \ \mbox{if}\ t \in [0,T], \ x \in (\Omega \setminus \cup_{i = 1}^N B_i(t)),\\ \\ 
\vu^S(t,x) \ \mbox{if}\ t \in [0,T], \ x \in \cup_{i = 1}^N B_i(t), 
\end{array}
\right.
\]
it holds:
 
{\bf (i)} 
the velocity $\vu^S$ is determined by \eqref{i6}, \eqref{i7} on each $B_i$, with the initial momentum 
\[
\vr^S(0, \cdot) \vu^S(0, \cdot) = \vm_0 \ \mbox{in}\ \Omega \setminus \Omega_F(0), 
\]
the initial position of the rigid bodies is $B_{i,0}$, 
\[
B_i(0) = B_{i,0},\ i = 1, \dots, N;
\]

\noindent
{\bf (ii)} the integral identity
\begin{equation} \label{M2}
\int_0^T \int_{\Omega_F(t)} \Big[ \vr^F \vu^F \cdot \partial_t \bfphi + 
\vr^F \vu^F \otimes \vu^F : \Grad \bfphi + p(\vr^F) \Div \bfphi \Big] \dx \dt = - \int_{\Omega_F(0)} \vm_0
\bfphi(0, \cdot) \ \dx
\end{equation}
holds for any $\bfphi \in C^1_{\rm loc}([0,T) \times \Ov{\Omega}; R^d)$, $\bfphi(t, \cdot) \cdot \vc{n}|_{\partial \Omega_F(t)} = 0$ 
for any $t \in [0,T)$; 

\noindent
{\bf (iii)} the Cauchy stress $\mathbb{T}$ in \eqref{i6}, \eqref{i7} takes the form
\begin{equation} \label{M3}
\mathbb{T} = - p(\vr^F) \mathbb{I}. 
\end{equation}
\end{itemize}

\end{Definition}

For the trace of $\vr^F$ in \eqref{M3} to exist, we need certain regularity of the fluid domain: (i) $\partial B_i$ at least Lipschitz 
for any $i = 1,\dots, N$, (ii) $B_i(t) \cap B_j(t) = \emptyset$ whenever $i \ne j$ and for any $t \in [0,T]$. Note that the latter is true at any $t$ if it is true at $t = 0$ as the governing velocity field $\vu^S$ is globally Lipschitz.

Next, we introduce the concept of admissible motion.

\begin{Definition}[{\bf Admissible motion}] \label{MD2}

The motion of the rigid bodies $\{ B^S_i \}_{i = 1}^N$ determined 
through the density $\vr^S$ and the velocity $\vu^S\in W^{1,\infty}([0,T]\times\Omega, \mathbb R^d)$ is \emph{admissible} if there exists 
$\vr^F \in C^2([0,T] \times \Ov{\Omega})$, $\inf_{(0,T) \times \Omega} \vr^F > 0$ such that 
\begin{equation} \label{M2a}
\int_{\Omega_F(t)} \vr^F (t,\cdot) \ \dx = m^F > 0 \ \mbox{for any}\ t \in [0,T],
\end{equation}
and \eqref{i6}, \eqref{i7} are satisfied with 
\[
\mathbb{T} = - p(\vr^F) \mathbb{I}.
\]

\end{Definition} 

\begin{Remark} \label{MR2}

Note that the density $\vr^F$ is determined by the pressure $p_F$ only on the boundaries $\partial B_i$ of the rigid bodies. Its extension 
to $\Omega$, and, in particular, to its fluid part $Q_F$ is completely arbitrary as soon as the total mass of the fluid is conserved - 
condition \eqref{M2a}.

\end{Remark}

Finally, we introduce the class of initial data compatible with an admissible motion. 

\begin{Definition}[{\bf Compatible initial data}] \label{MD3}

Let $\left[\vr^S, \vu^S, \{ B^S_i \}_{i = 1}^N \right]$ be an admissible rigid body motion with the associated density $\vr^F$ in the sense of 
Definition \ref{MD2}. We say that the initial data $\left[\vr_0, \vm_0, \left\{ B_{i,0}  \right\}_{i = 1}^N \right]$ are 
\emph{compatible} with the motion $\left[\vr^S, \vu^S, \{ B^S_i \}_{i = 1}^N \right]$ if:
\begin{itemize}
\item
\[
B_{i,0} = B^S_i(0),\ i=1, \dots, N,\ 
\vr_0 = \left\{ \begin{array}{l} \vr^S(0, \cdot) \ \mbox{in}\ \cup_{i = 1}^N B_{i,0},\\ \\
\vr^F (0, \cdot) \ \mbox{in} \ \Omega \setminus \cup_{i = 1}^N B_{i,0} \end{array} \right. ;
\]

\item 
\[
\vm_0 = \vr^S(0, \cdot) \vu^S(0, \cdot) \ \mbox{in}\ \cup_{i = 1}^N B_{i,0} ;
\]

\item
$\vm_0$ restricted to the fluid domain takes the form
\begin{equation} \label{M4}
\vm_0 = \Grad \Phi_0 \ \mbox{in}\ \Omega_F (0) \equiv \Omega \setminus \cup_{i = 1}^N B_{i,0},
\end{equation}
where 
\begin{equation} \label{M5}
- \Del \Phi_0 = \partial_t \vr^F(0, \cdot),\ \Grad \Phi_0 \cdot \vc{n}|_{\partial \Omega_F (0) } = 
\vr^F(0, \cdot) \vu^S(0, \cdot) \cdot \vc{n}|_{\partial \Omega_F (0)}.
\end{equation}

\end{itemize}

\end{Definition} 

\begin{Remark} \label{MR1}

The necessary compatibility conditions for solvability of the Neumann problem \eqref{M5} 
is a consequence of the transport theorem:
\begin{equation} \label{M6}
0 = \frac{{\rm d}}{{\rm d}t} \int_{\Omega_F(t)} \vr^F(t, \cdot) \ \dx = \int_{\Omega_F(t)} \partial_t \vr^F \ \dx 
+ \int_{\partial \Omega_F(t)} \vr^F \vu^S \cdot \vc{n} \ {\rm dS}_x 
\end{equation}
evaluated at $t = 0$.

\end{Remark}

\begin{Remark} \label{MR3}

As pointed out in Remark \ref{MR2}, the density profile $\vr_F$, in particular its initial value $\vr_F(0, \cdot)$, is fixed 
by the rigid bodies motion
only on the boundaries of the rigid bodies.

\end{Remark}

Our main goal is to show the following result:

\begin{Theorem}[{\bf Existence of weak solutions}] \label{MT1}

Let $\Omega \subset R^d$, $d = 2,3$ be a bounded domain of class $C^2$. Let $B_{i,0} = \Ov{S}_i$ be given, 
where $S_i \subset B_{i,0} \subset \Omega$ are simply connected domains of class $C^2$, $i = 1, \dots, N$. Suppose that $\left[ 
\vr^S, \vu^S, \{ B^S_i \}_{i = 1}^N \right]$ is an admissible motion in the sense of Definition 
\ref{MD2}. 

Then for any initial data $\left[\vr_0, \vm_0, \left\{ B_{i,0}  \right\}_{i = 1}^N \right]$ compatible with 
$\left[ 
\vr^S, \vu^S, \{ B^S_i \}_{i = 1}^N \right]$, the fluid--structure interaction problem \eqref{i1}--\eqref{i12} admits infinitely many weak solutions $\left[ \vr, \vm, \left\{ B_{i}  \right\}_{i = 1}^N \right]$ 
in the sense of Definition \ref{MD1}, where 
\[
B_i(t) = B_i^S(t) \ \mbox{for all}\ t \in [0,T], \ \mbox{and any}\ i= 1,\dots, N.
\]

\end{Theorem}

The rest of the paper is essentially devoted to the proof of Theorem \ref{MT1}. In view of the recent results by 
Chiodaroli \cite{Chiod},  De Lellis and Sz\' ekelyhidi \cite{DelSze3}, the conclusion of Theorem \ref{MT1} may not come as a complete surprise. The striking fact, however, is that solutions exist for any {\it a priori} given motion of the rigid objects. One may certainly argue that most of the solutions are not physical in the sense that the weak formulation does not include any kind of 
energy balance. Indeed for $\vu$ and $\varrho$ sufficiently smooth it is a routine matter to deduce from the momentum balance 
equations \eqref{i6}, \eqref{i7}, \eqref{i10} that
\[
\partial_t \mathcal E(t) = 0,
\]
where
\begin{equation}\label{M7a}
\mathcal E(t) = \int_{\Omega_F (t)} \frac12 \vr^F |\vu^F|^2 + P(\vr^F) \dx + \frac12 \sum_{i=1}^N \Big( 
m_i |\veta_i|^2 + \mathbb{J}^i: [\vomega_i \otimes \vomega_i] \Big).
\end{equation}
Here $P(\varrho)$ is the pressure potential,
\[
P'(\varrho)\varrho - P(\varrho) = p(\varrho).
\]
Solution are called \emph{admissible}, if the energy inequality set
\begin{equation} \label{M7}
\mathcal{E}(t) \leq \mathcal{E}(s) \ \mbox{holds for a.a.} \ 0 \leq s \leq t \leq T.
\end{equation}
As we shall see in the course of the proof of Theorem \eqref{MT1}, the solutions can be constructed to be admissible, at least in the 
open interval $(0,T)$.  
The crucial point, of course, is to see whether \eqref{M7} holds for $s = 0$, where the energy is expressed in terms of the initial data. We shall discuss this  issue in Section \ref{se:5}. 

\section{Reformulation}
\label{R}

We reformulate the problem to fit the abstract framework developed in \cite{Fei2016}. First, let us fix the admissible motion 
of the rigid bodies $\left[ \vr^S, \vu^S, \{ B_i \}_{i=1}^N \right]$ with the associated density $\vr^F$ as in Theorem \ref{MT1}. As $\vu^S$ is globally Lipschitz, we have 
\begin{equation} \label{R1}
B_i(t) \cap B_j(t) = \emptyset, \ B_i(t) \subset \Omega \ \mbox{for any}\ i=1,\dots,N, \ j \ne i,\ 
t \in [0,T]. 
\end{equation}
Accordingly, as the boundaries of $B_i$ are of class $C^2$, the part of $\Omega$ occupied by the fluid, 
\[
\Omega_F(t) = \Omega \setminus \cup_{i=1}^N B_i(t) \ \mbox{is a bounded domain of class}\ C^2 
\ \mbox{for any}\ t \in [0,T]. 
\]

Thus we may identify the potential $\Phi = \Phi(t, \cdot)$ at any $t \in [0,T]$ as the unique solution of the inhomogeneous Neumann problem:
\begin{equation} \label{R2}
- \Del \Phi (t) = \partial_t \vr^F(t, \cdot) \ \mbox{in}\ 
\Omega_F(t), \ \Grad \Phi(t) \cdot \vc{n}|_{\partial \Omega_F(t)} = 
\vr^F(t, \cdot) \vu^S(t, \cdot) \cdot \vc{n}|_{\partial \Omega_F(t) },
\end{equation}
normalized by the condition 
\[
\int_{\Omega_F(t)} \Phi(t) \ \dx = 0.
\]

We fix the density $\vr$, 
\[
\vr(t, \cdot) = \left\{ \begin{array}{l} \vr^S(t,\cdot) \ \mbox{in}\ \Omega_S (t), \\ \\ 
\vr^F(t, \cdot) \ \mbox{in}\ \Omega_F (t), \end{array} \right. \ t \in [0,T].           
\]
The velocity $\vu^F$ in the fluid part will be determined via the momentum $\vm^F = \vr_F \vu_F$, where 
\[
\vm^F = \Grad \Phi + \vc{v}, \ \mbox{where}\ \Div \vc{v} = 0, \ \vc{v}(t, \cdot) \cdot \vc{n}|_{\partial \Omega_F(t)} = 0
\ \mbox{for any}\ t \in [0,T]. 
\]
In the weak sense, the conditions imposed on $\vc{v}$ may be stated as 
\begin{equation} \label{R3}
\int_0^T \int_{\Omega_F(t)} \vc{v} \cdot \Grad \varphi \ \dxdt = 0
\ \mbox{for any}\ \varphi \in C^1([0,T] \times R^d).
\end{equation}
Now, observe that \eqref{R2}, \eqref{R3} imply that the equation of continuity \eqref{M1} is automatically satisfied. 

Thus the proof of Theorem \ref{MT1} reduces to finding a function $\vc{v}$ satisfying the momentum equation in the 
fluid domain:
\[
\partial_t \vc{v} + \Div \left( \frac{ (\vc{v} + \Grad \Phi) \otimes (\vc{v} + \Grad \Phi) }{\vr^F} \right) 
+ \Grad \left( p(\vr_F) + \partial_t \Phi \right) = 0 , \ \vc{v}(0, \cdot) = 0.
\]
In view of \eqref{R2}, \eqref{R3}, this can be written in the weak form:
\begin{equation} \label{R4}
\int_0^T \int_{\Omega_F(t)} \left[ \vc{v} \cdot \partial_t \bfphi 
+ \frac{ (\vc{v} + \Grad \Phi) \otimes (\vc{v} + \Grad \Phi) }{\vr^F} : \Grad \bfphi + 
\left( p(\vr_F) + \partial_t \Phi \right)  \Div \bfphi \right] \dxdt  
= 0
\end{equation}
for any $\bfphi \in C^1_c([0,T) \times R^d; R^d)$, $\bfphi (t, \cdot) \cdot \vc{n} |_{\partial \Omega_F(t)} = 0$.

We infer that the proof of Theorem \ref{MT1} reduces to finding the field $\vc{v} \in L^\infty(Q_F; R^d)$ satisfying 
\eqref{R3}, \eqref{R4} for given 
\begin{equation} \label{R5}
\vr^F ,\ \Phi, \ \partial_t \Phi, \ (\vr^F)^{-1} \in C(\Ov{Q}_F).
\end{equation}

\section{Convex integration}
\label{C}

The problem \eqref{R3}, \eqref{R4}, with fixed parameters satisfying \eqref{R5}, may be solved by a version of the convex integration method developed in \cite{Fei2016}. We start by rewriting the equations in a slightly different form: 
\begin{equation} \label{C1}
\int_0^T \int_{\Omega_F(t)} \vc{v} \cdot \Grad \varphi \ \dxdt = 0
\ \mbox{for any}\ \varphi \in C^1([0,T] \times R^d), 
\end{equation}
\begin{equation} \label{C2}
\int_0^T \int_{\Omega_F(t)} \left[ \vc{v} \cdot \partial_t \bfphi 
+ \left( \frac{ (\vc{v} + \Grad \Phi) \otimes (\vc{v} + \Grad \Phi) }{\vr^F} -  
\frac{1}{d} \frac{|\vc{v} + \Grad \Phi |^2}{\vr^F} \mathbb{I} \right)  : \Grad \bfphi \right] \dxdt  
= 0
\end{equation}
for any $\bfphi \in C^1_c([0,T) \times R^d)$. In addition, we prescribe the kinetic energy, 
\begin{equation} \label{C3}
\frac{1}{2} \frac{|\vc{v} + \Grad \Phi |^2}{\vr^F} = - \frac{d}{2} \Big( \partial_t \Phi + p(\vr_F) \Big) + \Lambda(t) 
\equiv E(t) \ \mbox{a.a. in}\ Q_T,
\end{equation}
where $\Lambda = \Lambda(t)$ is a spatially homogeneous function to be determined below. Observe that \eqref{C2}, \eqref{C3} 
yield \eqref{R4} as soon as the test function $\bfphi$ in \eqref{C2} satisfies 
\[
\bfphi (t, \cdot) \cdot \vc{n}|_{\partial \Omega_F(t)} = 0.
\]

To solve \eqref{C1}--\eqref{C3}, we first follow the strategy of De Lellis and Sz\' ekelyhidi introducing the space of subsolutions 
$X_0$ containing velocity fields $\vc{v}$ with the associated fluxes $\mathbb{F}$ satisfying:
\begin{itemize}
\item
\[
\vc{v} \in C^1_c(Q_T; R^d),\ \mathbb{F} \in C^1_c({Q}_T; R^{d \times d}_{{\rm sym,0}});
\]
\item
\[
\Div \vc{v} = 0,\ 
\partial_t \vc{v} + \Div \mathbb{F} = 0\ \mbox{in} \ R^{d + 1}
\]
\item
\[
\frac{1}{2} \frac{|\vc{v} + \Grad \Phi|^2}{\vr^F} \leq 
\frac{d}{2} \lambda_{\rm max} \left[ \frac{ (\vc{v} + \Grad \Phi) \otimes (\vc{v} + \Grad \Phi)}{\vr^F} - 
\mathbb{F} \right] < E \in Q_T.
\] 

\end{itemize}

As explained in detail in \cite[Section 2--4]{Fei2016}, the solutions of \eqref{C1}--\eqref{C3} are obtained as zero points of the convex functional 
\[
J[\vc{v}] = \int_0^T \int_{\Omega_F(t)} \left[ \frac{1}{2} \frac{|\vc{v} + \Grad \Phi|^2}{\vr^F} - E \right] \dxdt 
\]
defined on a completion $X$ of the space of subsolutions $X_0$ with respect to the metrics of the topological space 
\[
C_{{\rm weak}}([0,T]; L^2(\Omega; R^d)) \cap L^\infty((0,T) \times \Omega; R^d).
\] 
As shown in \cite{Fei2016}, $J$ vanishes on the set of its points of continuity, where the latter is not of the first Baire category 
in $X$, in particular, it is dense in $X$, see \cite{Fei2016} for details. 

Accordingly, it is enough to show that $X$ is non--trivial, meaning the set of subsolutions $X_0$ is non--empty. 
Considering the trivial ansatz $\vc{v} = 0$, $\mathbb{F} = 0$, we get 
\begin{equation} \label{C4}
\frac{d}{2} \lambda_{\rm max} \left[ \frac{\Grad \Phi \otimes \Grad \Phi}{\vr^F}  \right] < \Lambda 
- \frac{d}{2} \partial_t \Phi - \frac{d}{2} p(\vr_F)
\end{equation}
As $\Phi$, $\vr^F$ are fixed belonging to the class \eqref{R5}, one can definitely find $\Lambda = \Lambda(t)$ so that
\eqref{C4} holds in $\Ov{Q}_T$. This completes the proof of Theorem \ref{MT1}. 

Finally, observe that the total energy now reads 
\[
\begin{split}
\mathcal{E}(t) &= \int_{\Omega_F(t)} \frac{1}{2} \frac{ |\vc{v} + \Grad \Phi|^2 }{\vr^F} + P(\vr_F) \ \dx 
+ \frac12 \sum_{i=1}^N \Big( 
m_i |\veta_i|^2 + \mathbb{J}^i: [\vomega_i \otimes \vomega_i] \Big)\\
&= \Lambda(t) |\Omega_F(t)| + \int_{\Omega_F(t)} \left[ P(\vr^F) - \frac{d}{2} \partial_t \Phi  - \frac{d}{2} p(\vr^F) \right] \dx
+ \frac12 \sum_{i=1}^N \Big( 
m_i |\veta_i|^2 + \mathbb{J}^i: [\vomega_i \otimes \vomega_i] \Big),
\end{split}
\]
where the last equality holds for a.a. $t \in [0,T]$. In particular, we can choose $\Lambda = \Lambda(t)$ in such a way that the 
total energy equals for a.a. $t \in (0,T)$ to a strictly decreasing function, meaning the solutions are admissible in the sense of 
\eqref{M7}. Thus we have obtained the following. 

\begin{Corollary}[{\bf Admissible solutions}] \label{CC1}

Under the hypotheses of Theorem \ref{MT1}, the fluid--structure interaction problem \eqref{i1}--\eqref{i12} admits infinitely 
many weak solutions $\left[ \vr, \vm, \left\{ B_{i}  \right\}_{i = 1}^N \right]$ 
in the sense of Definition \ref{MD1}, where 
\[
B_i(t) = B_i^S(t) \ \mbox{for all}\ t \in [0,T], \ \mbox{and any}\ i= 1,\dots, N.
\]
In addition, the solutions satisfy the energy inequality 
\begin{equation} \label{C5}
\mathcal{E}(t) \leq \mathcal{E}(s) \ \mbox{for a.a.}\ 0 \leq s \leq t \leq T.
\end{equation}

\end{Corollary}

The question when the energy inequality \eqref{C5} includes the initial time $s = 0$ will be discussed in the next section. 

\section{Energy inequality, physically relevant solutions}
\label{P}\label{se:5}

We briefly discuss the validity of the energy inequality ``up to the origin'', specifically 
\begin{equation} \label{P1}
\begin{split}
\mathcal{E}(t) &= \int_{\Omega_F(t)} \left[ \frac{1}{2} \frac{ |\vm|^2 }{\vr} + P(\vr) \right](t, \cdot) \dx 
+ \frac12 \sum_{i=1}^N \Big( 
m_i |\veta_i(t)|^2 + \mathbb{J}^i: [\vomega_i(t) \otimes \vomega_i(t)] \Big)\\ &\leq
\int_{\Omega_F(0)} \left[ \frac{1}{2} \frac{ |\vm_0|^2 }{\vr_0} + P(\vr_0) \right] \dx 
+ \frac12 \sum_{i=1}^N \Big( 
m_i |\veta_{0,i}|^2 + \mathbb{J}^i: [\vomega_{0,i} \otimes \vomega_{0,i}] \Big),\ t \geq 0.
\end{split}
\end{equation}
In view of the specific construction used in Section \ref{C}, notably with $\Lambda$ satisfying \eqref{C4}, 
we do not expect \eqref{P1} to hold for the solutions obtained in Section \ref{C}. However, Corollary \ref{CC1} 
ensures the existence of a full measure sets of times $s$ such that \eqref{C5} holds. This yields the following result.

\begin{Theorem}[{\bf Existence of weak solutions satisfying energy inequality}] \label{PT1}

Let $\Omega \subset R^d$, $d = 2,3$ be a bounded domain of class $C^2$. Let $B_{i,0} = \Ov{S}_i$ be given, 
where $S_i \subset B_{i,0} \subset \Omega$ are simply connected domains of class $C^2$, $i = 1, \dots, N$. Suppose that $\left[ 
\vr^S, \vu^S, \{ B^S_i \}_{i = 1}^N \right]$, with the associated fluid density $\vr^F$, is an admissible motion in the sense of Definition \ref{MD2}.

The there exists a set $\mathcal{S} \subset (0,T)$ of full Lebesgue measure such that for any $s \in \mathcal{S}$ we have the following property: 

For any data 
\[
B_{i,0} = B^S_i(s), \ i = 1,\dots, N,\ 
\vr_0 = \left\{ \begin{array}{l} \vr^S(s) \ \mbox{in}\ \cup_{i=1}^N B_{i,0} ,\\ \\ 
\vr^F(s) \ \mbox{in}\ \Omega \setminus \cup_{i=1}^N B_{i,0} \end{array} \right., s \in \mathcal{S}
\] 
there exist (infinitely many) $\vm_0$ such that 
the fluid--structure interaction problem \eqref{i1}--\eqref{i12} admits a weak solution 
$\left[ \vr, \vm, \left\{ B_{i}  \right\}_{i = 1}^N \right]$ 
in the sense of Definition \ref{MD1} in $(0, T-s)$. In addition
\[
B_i(t) = B_i^S(t+s) \ \mbox{for all}\ t \in [0,T-s], \ \mbox{and any}\ i= 1,\dots, N,
\]
and the energy inequality \eqref{P1} holds.

\end{Theorem}

Theorem \ref{PT1} asserts the existence of at least one solution, however, a refined analysis performed in \cite[Section 6]{Fei2016} may be used to 
show that there are in fact infinitely many solutions starting from the same initial data. We leave the interested reader to work out the details.
 
\section{On admissible motion}
\label{pres.con}

We conclude the paper by several examples concerning possible admissible motions of rigid bodies.

\subsection{A ball in three dimensions}

Let $\bfeta$ be a given velocity of the center of gravity of some ball $B$. We have, according to \eqref{i6},
$$
\pat \bfeta = \frac1{m_B}\int_{\partial B}p \vn \ {\rm d}S_x.% = \frac1{m_B}\int_B \nabla p \ \dx.
$$
We consider a pressure 
$$
p = \frac{m_B}{|B|} \pat\bfeta(t)\cdot x \sigma_{U} + p_0(t)
$$
where $U$ is an open neighborhood of $B$ and $\sigma:\Omega \mapsto [0,\infty)$ is a smooth function satisfying $\chi_B\leq \sigma \chi_U$. The function $p_0:[0,T]\mapsto \mathbb R$ is such that the whole pressure is positive and the corresponding density $\varrho$ satisfies \eqref{M2a} for every $t$. Such pressure may also induce some rotation -- see Remark \ref{hetero.ball}. Anyway, we may construct an admissible pressure and density for any given smooth translation of the center of gravity of the given ball.

Let now assume the center of gravity agrees with the geometrical center of the ball, i.e.
$$
\frac 1{m_B} \int_B \varrho(x) x \ \dx = \frac {1}{|B|} \int_B x \ \dx.
$$
Then the rotation of such ball cannot be influenced by an action of the perfect fluid. Indeed, the second term on the right hand side of \eqref{i7} is zero for every sufficiently smooth pressure. This is a consequence of $$(x-x_{B})\times \vn = 0,\ \mbox{for every }x\in\partial B.$$
Consequently, any smooth translation of such ball is an admissible motion assuming the rotation is constant in time.
The same applies to a finite number of rigid balls provided there is neiter mutual contact nor a contact with the boundary $\partial \Omega$ at the initial time.

\subsection{Homogeneous body in two dimensions}

We assume $\varrho(x)|_{B_i} =1$. Throughout this section we understand $\vn:\partial B\mapsto S$ as a function which assign a normal direction to a point of a boundary. Moreover, we assume $i$ is fixed through this chapter and thus we drop this particular index.
\begin{Lemma}\label{forces}
Assume $B$ is a strictly convex compact body with a $C^2$ boundary. Then for every $x_0\in \partial B$ and every neighborhood $U_{x_0}$ there exists a bounded $C^1$ function $p:\mathbb R^2\mapsto \mathbb R^+_0$ supported on $U_{x_0}$ such that 
\begin{equation}\label{integral.p}
\int_{\partial B} p {\bf n}\ {\rm d}S_x = {\bf F},
\end{equation}
where $\vF$ has a direction of $\vn(x_0)$.
\end{Lemma}
\begin{proof}
There is a neighborhood of $x_0 = 0$   and a coordinate system such that there exists a function $h(x):\mathbb R\mapsto \mathbb R^+_0$, $h(0) = h'(0) = 0$, $h(x)\neq 0$ for $x\neq 0$,  $\{(x,h(x))\} = \partial B$ on the neighborhood and, moreover, $${\bf n} = \frac{(h'(x),1)}{\sqrt{1+h'^2(x)}}.$$
The choice of a coordinate system yields $\vF$ is of the form $(0,F_2)$ with $F_2>0$. In what follows, we assume $U_{x_0}$ is contained in this neighborhood and we look for  a positive function $p$ supported in $U_{x_0}$. The right hand side of \eqref{integral.p} written in coordinates has a form
$$
\left(\int_{-R}^R p(x)H(x)\ \dx, \int_{-R}^R \frac{p(x)}{\sqrt{1+h'^2(x)}}\ \dx\right),
$$
where $R$ is such that $(-R,R)\times \{0\}\subset U_{x_0}$. Here $H(x) = \frac{h'(x)}{\sqrt{1+h'^2(x)}}$. Due to assumptions, $H(x) x >0$ for all $x\neq 0$.
Note that the second integral gives some positive value, denote it by $c$. 

It is possible to find a function $p$ such that the first coordinate is $0$. Indeed, Let $\mathcal P$ denote some even bounded smooth non-negative function with support in $(-R,R)$. Let $f$ be defined as
$$
f(x) =  \sqrt{-\frac{H(-x)}{H(x)}}.
$$
Due to assumptions, $f$ can be extended continuously to $0$. The function $p(x) = \mathcal P(x)f(x)$ posses all demanded qualities. Indeed, it is bounded, continuous and, due to the definition of $f$, $p(x)H(x)$ is an odd function. Furthermore, since $\mathcal P$ and  $f H$ are at least $C^1$ functions, it follows that $p$ is also a $C^1$ function. It remains to multiply this function by $\frac{F_2}{c}$ to get the demanded value of the second coordinate.

The function $p$ is now defined on a boundary of $B$. However, it is trivial to extend it on the whole $\mathbb R^2$ such that its support remains in a neighborhood of $x_0$. 
\end{proof}

\begin{Lemma}\label{pres.rotation}
Let $B$ fulfill assumptions of the previous lemma and let it be not a ball. Then there exists a non-negative pressure $p$ such that 
\begin{equation*}
\begin{split}
\int_{\partial B} (x-x_{B})\cdot p \vc{n}^\bot \ \dx & =: \bfom\neq 0,\\
\int_{\partial B} p {\bf n}\ {\rm d}S_x &= 0.
\end{split}
\end{equation*}

\end{Lemma}

\begin{proof}
We recall that $\vn:\partial B\mapsto S$ assigns a normal direction to a point of the boundary. Let $N:S\mapsto \partial B$ be its inverse -- note it is well defined due to the assumption about the convexity of the body. Furthermore, we define 
\begin{equation*}
\begin{split}
T:S&\mapsto \mathbb R\\
T:\vm &\mapsto (N(\vm) - x_B)\cdot \vm^\bot
\end{split}
\end{equation*}

Let \begin{equation}\label{ass.force}\vm\in S \mbox{ be such that }T(\vm)\mbox{ and }T(-\vm)\mbox{ have the same sign.}\end{equation}
Assume it is positive. Due to the smoothness of $\partial B$ there exist nonempty neighborhoods $U_{\vm},\ U_{-\vm}\subset \partial B$ such that $(x-x_B)\cdot \vn(x)^\bot$ is positive for $x\in U_{\vm}\cup U_{-\vm}$. According to Lemma \ref{forces} there exist (positive) pressures creating forces in direction $\vm$ and $-\vm$. By a proper linear combination with positive coefficients it is possible to construct a pressure supported in $U_{\vm}\cup U_{-\vm}$ such that the resulting force is $0$. However, due to the positivity of $(x-x_B)\cdot \vn(x)$, the resulting $\bfom$ is positive. Note that it is sufficient to have $T(\vm) \neq 0$ once we know that $(x-x_B)\cdot \vn(x)^\bot = 0$ for $x\in \partial B$ form some neighborhood of $N(\vm)$.

%\begin{figure}
%\centering
%\includegraphics[width=6cm]{sphere.pdf}
%\caption{Sphere}
%\label{fig:sphere}
%\end{figure}

Assume now that $B$ is such that 
\begin{equation}\label{ass.force.2} T(\vm_0)=0\mbox{ implies }T(\vm)\neq 0\mbox{ for }\vm\in U_{\vm_0}\setminus \vm_0\mbox{ for some neighborhood }U_{\vm_0}\subset S.\end{equation}
 Further, for $\vm_0\in S$ in form $\vm = (\cos \theta_0,\sin\theta_0),\ \theta_0\in \mathbb T_{[0,2\pi)}$ we establish right and left neighborhood as follows:
$$U_{\vm_0}^+ := \{(\cos\theta,\sin\theta), \theta\in (\theta_0,\theta_0+\varepsilon)\subset\mathbb T_{[0,2\pi)}\}\mbox{ for some }\varepsilon\in(0,\pi]$$ 
and
$$U_{\vm_0}^- = \{(\cos\theta,\sin\theta), \theta\in (\theta_0-\varepsilon,\theta_0)\subset\mathbb T_{[0,2\pi)}\}\mbox{ for some }\varepsilon\in (0,\pi].$$
There exists $\vm_0$ such that $T(\vm_0)=0$ and $T(\vm)>0$ for $\vm\in U_{\vm_0}^+$ and $T(\vm)<0$ for $\vm\in U_{\vm_0}^-$.  Assume $T(-\vm_0)= 0$ and $T(\vm)<0$ for $\vm \in U_{\vm_0}^+$ and $T(\vm)>0$ for $\vm\in U_{\vm_0}^-$, otherwise there can be found $\vm$ fulfilling assumption \eqref{ass.force} and we are done. We distinguish two cases:
\begin{itemize}
\item $T(\vm)$ is non-negative in the whole arc connecting $\vm$ and $-\vm$, in other words, the neighborhood $U_{\vm_0}^+$ fulfilling $T(\vm)\geq0$ $\forall \vm \in U_{\vm_0}^+$ can be chosen in such a way that $U_{\vm_0}^+ = U_{-\vm_0}^-$. In  that case, either $\{\vm, T(\vm)>0\}\cap (S\setminus U_{\vm_0}^+)\neq \emptyset$ and we can chose $\vm$ such that \eqref{ass.force} is fulfilled, or $\{\vm, T(\vm)>0\} \cap (S\setminus U_{\vm_0}^+)=\emptyset$ -- however, this leads to a contradiction with the definition of $x_B$.
\item There exists $\vm_1\in U_{\vm_0}^+\cap U_{\vm_0}^-$ such that $T(\vm) <0$ for $\vm\in U_{\vm_1}$. If $T(\vm)<0$ for $\vm \in U_{-\vm_1}$, we are done. Otherwise, there exists at least three open sets $U^+_{\vm_0}$, $U^-_{-\vm_0}$ and $U_{-\vm_1}$ on which $T(\vm)$ is positive. We chose one vector from each of these sets, denoting them $\vm_a$, $\vm_b$, $\vm_c$. We use Lemma \ref{forces} to construct non-negative pressures which create forces in direction $\vm_a$, $\vm_b$, and $\vm_c$. Their linear combination with positive coefficients gives the resulting force equal to $0$. However, the resulting $\bfom$ is positive.
\end{itemize}

The lemma can be proven similarly even if \eqref{ass.force.2} is not fulfilled, the method of the proof is the same. The assumption just allows to keep the notation lucid. 

\end{proof}

\begin{Remark}
The pressures constructed by the above Lemmas are smooth with respect to $F$ and $\bfom$. Moreover, they are $C^1$ with respect to a position of $B$
\end{Remark}

\begin{Remark}
Let $x_0\in \partial B$. The pressure having all properties of Lemma \ref{pres.rotation} can be constructed under additional constraint $\supp p\cap U_{x_0} = \emptyset$ for some neighborhood of $x_0$. This follows easily from the smoothness of $\partial B$. Indeed, if one of the pressures from the proof of the lemma should be supported on a neighborhood of $x_0$, the pressure can be, due to the smoothness, shifted slightly aside.
\end{Remark}

\begin{Remark}
In case $B$ is a homogeneous ball, the fluid cannot affect the rotation of $B$. Indeed, the integral in \eqref{i7} is always zero.
\end{Remark}

Now we can proceed to a construction of pressure. Let the movement of a body can be explained by a force $\vF(t)$ and a torgue $\bfom(t)$. Fix $t\in [0,T]$. By Lemma \ref{forces} there can be found pressure $p_1$ such that 
$$
\int_{\partial B} p_1\vn {\rm d}S_x = \vF(t).
$$
Further, there exists $p_2$ such that $\int_{\partial B} p_2\vn {\rm d}S_x = 0$ and simultaneously 
$$
\int_{\partial B} p_2 (x-x_B)\cdot \vn^\bot {\rm d}S_x =\bfom_1\neq 0.
$$
It suffices to take $$p_3 = \left(p_1 +  \left(\frac{\bfom(t)}{\bfom_1} - \frac{\int_{\partial B} p_1(x-x_B)\cdot \vn^\bot {\rm d}S_x}{\bfom_1}\right)p_2\right)\sigma_U + p_0(t).$$
Similarly as in the previous section, $U\supset B$ is a neighborhood of $B$, $\sigma_u$ is a smooth function satisfying $\chi_B\leq \sigma_U\leq \chi_U$ and $p_0(t)$ is such that the total pressure $p_3$ is positive and the corresponding density satisfies \eqref{M2a}. 

\begin{Remark} \label{hetero.ball}The assumption that body is homogeneous cannot be omitted. Indeed, concern a body $$B = \{(x_1,x_2)\in\mathbb R^2, (x_1-1)^2 + x_2^2 \leq 4\}$$
with a center of gravity $x_B = 0$. It follows that for $x\in \partial B$ we have $\vn = \frac12 (-x_1 + 1, -x_2)$ and $x\cdot \vn^\bot = \frac{x_2}2$. Thus, for every pressure inducing a force $\vF = (0,1)$ we have
$$
1 = -\int_{\partial B} p \frac{x_2}2\ {\rm d}S_x = \int_{\partial B} p x\cdot \vn^\bot\ {\rm d}S_x = \bfom
$$
and, consequently, every force in the direction of $(0,1)$ induces also some nontrivial torque. 
\end{Remark}

\def\cprime{$'$} \def\ocirc#1{\ifmmode\setbox0=\hbox{$#1$}\dimen0=\ht0
  \advance\dimen0 by1pt\rlap{\hbox to\wd0{\hss\raise\dimen0
  \hbox{\hskip.2em$\scriptscriptstyle\circ$}\hss}}#1\else {\accent"17 #1}\fi}

%\bibliography{citace}

\begin{thebibliography}{1}

\bibitem{Chiod}
E.~Chiodaroli.
\newblock A counterexample to well-posedness of entropy solutions to the
  compressible {E}uler system.
\newblock {\em J. Hyperbolic Differ. Equ.}, {\bf 11}(3):493--519, 2014.

\bibitem{DelSze3}
C.~De~Lellis and L.~Sz{\'e}kelyhidi, Jr.
\newblock On admissibility criteria for weak solutions of the {E}uler
  equations.
\newblock {\em Arch. Ration. Mech. Anal.}, {\bf 195}(1):225--260, 2010.

\bibitem{EVGA}
L.~C.~Evans and R.~F.~Gariepy
\newblock {\em Measure theory and fine properties of functions.}
\newblock CRC Press, Boca Raton 1992.
  

\bibitem{Fei2016}
E.~Feireisl.
\newblock Weak solutions to problems involving inviscid fluids.
\newblock In {\em Mathematical Fluid Dynamics, Present and Future}, volume 183
  of {\em Springer Proceedings in Mathematics and Statistics}, pages 377--399.
  Springer, New York, 2016.

\bibitem{galdi}
G. P. Galdi.
\newblock On the motion of a rigid body in a viscous liquid: a mathematical analysis with applications.
\newblock In {\em Handbook of mathematical fluid dynamics}, Vol. 1, 653--791, North-Holland, Amsterdam, 2002

\bibitem{HoMaTu}
J. G. Houot and J. San Martin and M. Tucsnak.
\newblock Existence of solutions for the equations modeling the motion of rigid bodies in an ideal fluid.
\newblock {\em Journal of Functional Analysis}, 259, 2856--2885, 2010

\bibitem{OrRoTa} 
J. Ortega and L. Rosier and T. Takahashi.
\newblock On the motion of a ridgid body in a bidimensional incompressible perfect fluid.
\newblock {\em Ann. l. H. Poincar\'e} AN24, 139--165, 2007


\end{thebibliography}
%\bibliographystyle{plain}

\end{document}